%% file: neurips_2026.tex
\documentclass{article}

 \usepackage[preprint]{neurips_2026}

\usepackage[utf8]{inputenc} 
\usepackage[T1]{fontenc}    
\usepackage{hyperref}       
\usepackage{url}            
\usepackage{booktabs}       
\usepackage{amsfonts}       
\usepackage{nicefrac}       
\usepackage{microtype}      
\usepackage{xcolor}         

\usepackage{bbm}
\usepackage{enumerate}
\usepackage{algorithm}
\usepackage{algorithmicx}
\usepackage{algpseudocodex}
\usepackage{enumitem}
\algrenewcommand\algorithmicrequire{\textbf{Input:}}
\algrenewcommand\algorithmicensure{\textbf{Output:}}

\usepackage{amsmath}
\usepackage{xcolor}
\definecolor{darkgreen}{rgb}{0,0.5,0}
\usepackage{hyperref}
\hypersetup{
    unicode=false,          
    colorlinks=true,        
    linkcolor=red,          
    citecolor=darkgreen,    
    filecolor=magenta,      
    urlcolor=blue           
}
\usepackage[capitalize, nameinlink]{cleveref}

\usepackage{amsthm}
\usepackage{thm-restate}
\theoremstyle{plain}
\newtheorem{theorem}{Theorem}

\newtheorem{proposition}[theorem]{Proposition}
\theoremstyle{definition}
\newtheorem{definition}[theorem]{Definition}

\theoremstyle{remark}

\crefname{assumption}{Assumption}{Assumptions}
\makeatletter
\AddToHook{cmd/appendix/before}{\def\cref@section@alias{appendix}}
\makeatother

\usepackage{tikz}
\usetikzlibrary{calc, graphs, graphs.standard, shapes, arrows, arrows.meta, positioning, decorations.pathreplacing, decorations.markings, decorations.pathmorphing, fit, matrix, patterns, shapes.misc, tikzmark}

\newcommand{\eps}{\varepsilon}

\title{Network-Based Interventions for HIV Prevention via Cascade-Aware Suppression of Transmission}

\author{
Akseli Kangaslahti\thanks{Equal contribution}\\
Harvard University\\
\texttt{akselikangaslahti@g.harvard.edu}
\And
Davin Choo\footnotemark[1]\\
Harvard University\\
\texttt{davinchoo@seas.harvard.edu}
\And
Milind Tambe\\
Harvard University\\
\texttt{tambe@seas.harvard.edu}
\And
Alastair van Heerden\\
University of Witwatersrand\\
Wits Health Consortium\\
\texttt{alastair.vanheerden@wits.ac.za}
\And
Cheryl Johnson\\
World Health Organization\\
\texttt{johnsonc@who.int}
}

\begin{document}

\maketitle

\begin{abstract}
Treating and preventing Human Immunodeficiency Virus (HIV) remains a critical global health challenge. While antiretroviral therapy provides a path toward viral suppression --- effectively eliminating an individual's transmission risk --- systemic resource constraints limit the reach of intervention efforts. This work addresses the strategic distribution of intensive resources among virally unsuppressed individuals to minimize the expected cascade of new infections within a transmission network. We formalize this challenge as a novel constrained optimization problem where we have resources to ``treat'' $k$ out of a set $\mathbf{P}$ of virally unsuppressed individuals, and establish its theoretical connections to existing computational literature. We then propose Cascade-Aware Suppression of Transmission (\texttt{CAST}), a polynomial-time $(\delta, \epsilon)$-approximation algorithm that achieves a $2\sqrt{|\mathbf{P}|}$ approximation ratio by leveraging connections to the Minimum-$k$-Union (\texttt{MkU}) problem and Hoeffding-style concentration bounds. Extensive evaluations on real-world HIV networks demonstrate that \texttt{CAST} outperforms standard public health and computer science baselines. Furthermore, we show that \texttt{CAST} is empirically robust across diverse infectious disease networks, varied edge probability initializations, and settings involving imperfect network data.
\end{abstract}

\section{Introduction}
\label{intro}

With more than 40 million individuals infected worldwide, Human Immunodeficiency Virus (HIV) remains a critical global health challenge.
While a cure is currently unavailable, proper medical treatment can suppress the virus to levels that effectively eliminate further transmission risk. 
Aligned with UN Sustainable Development Goal 3.3 \cite{UN_SDG3}, UNAIDS proposed the 95-95-95 targets for HIV that aim for 95\% of individuals with HIV to be aware of their status, 95\% of those aware of their positive status to receive treatment, and 95\% of those that are treated to reach viral suppression by 2030 \cite{unaids2022}.
However, public health agencies operate under chronic resource constraints, so optimizing resource allocations for HIV treatment and prevention is essential \cite{sansom2021optimal}.
This is an especially urgent problem as public health agencies that are already constrained are currently facing additional budget cuts for HIV prevention efforts \cite{ten2025impact}.
The World Health Organization (WHO) has recommended network-based strategies to help navigate these resource constraints \cite{who}.

In this work, we address one such resource allocation challenge to support network-based HIV prevention efforts.
Clinical studies have established that virally suppressed individuals do not transmit the virus \cite{cohen2011prevention,rodger2016sexual,bavinton2018viral}, a finding that anchors the ``Undetectable = Untransmittable'' (U = U) campaign \cite{niaid2019hiv,olson2020uu,CDC_UU_2024} and underpins the public health strategy known as Treatment as Prevention (TasP), which leverages individual treatment as a structural tool to reduce population-level incidence \cite{DHS_HIV}.
However, viral suppression can be hindered by a variety of factors, including medication non-adherence and limited healthcare access, which in turn contribute to broader drug resistance challenges.
Though public health agencies can provide some support, such as adherence counseling, viral load monitoring, or high-efficacy treatments, resources are limited.
Thus, when a public health agency identifies a group of individuals that have untreated HIV and/or are virally unsuppressed, they must choose how to strategically allocate these resources to prioritize a limited number of individuals for intensive interventions, with the goal of minimizing the expected cascade of new infections across the community's transmission network.
For example, this scenario can arise when trying to expand services to reach underserved populations.
Furthermore, since collecting information to understand the full network can be expensive and time consuming, the intervention strategy should ideally adapt well to imperfect networks that are less costly and faster to collect.
We focus on developing an algorithmic approach that can help address this problem.
We conduct this research in close collaboration with domain experts from a large university in South Africa as well as a global nonprofit health organization, with the ultimate goal of deploying our approach in field settings.
\cref{fig:field-work} shows HIV prevention field work being conducted by our collaborators.

\begin{figure}
    \centering
    \includegraphics[height=65pt]{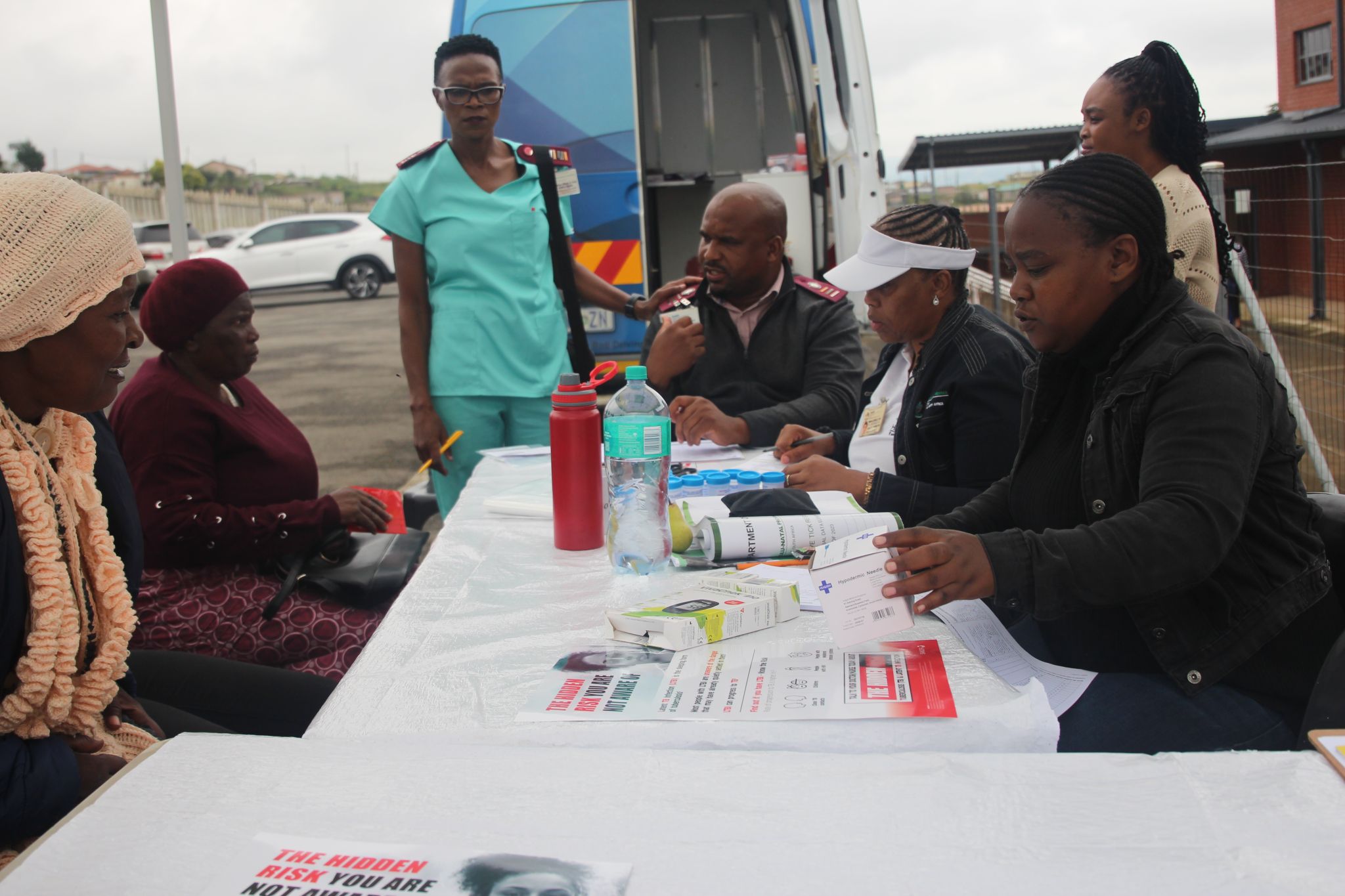}
    \includegraphics[height=65pt]{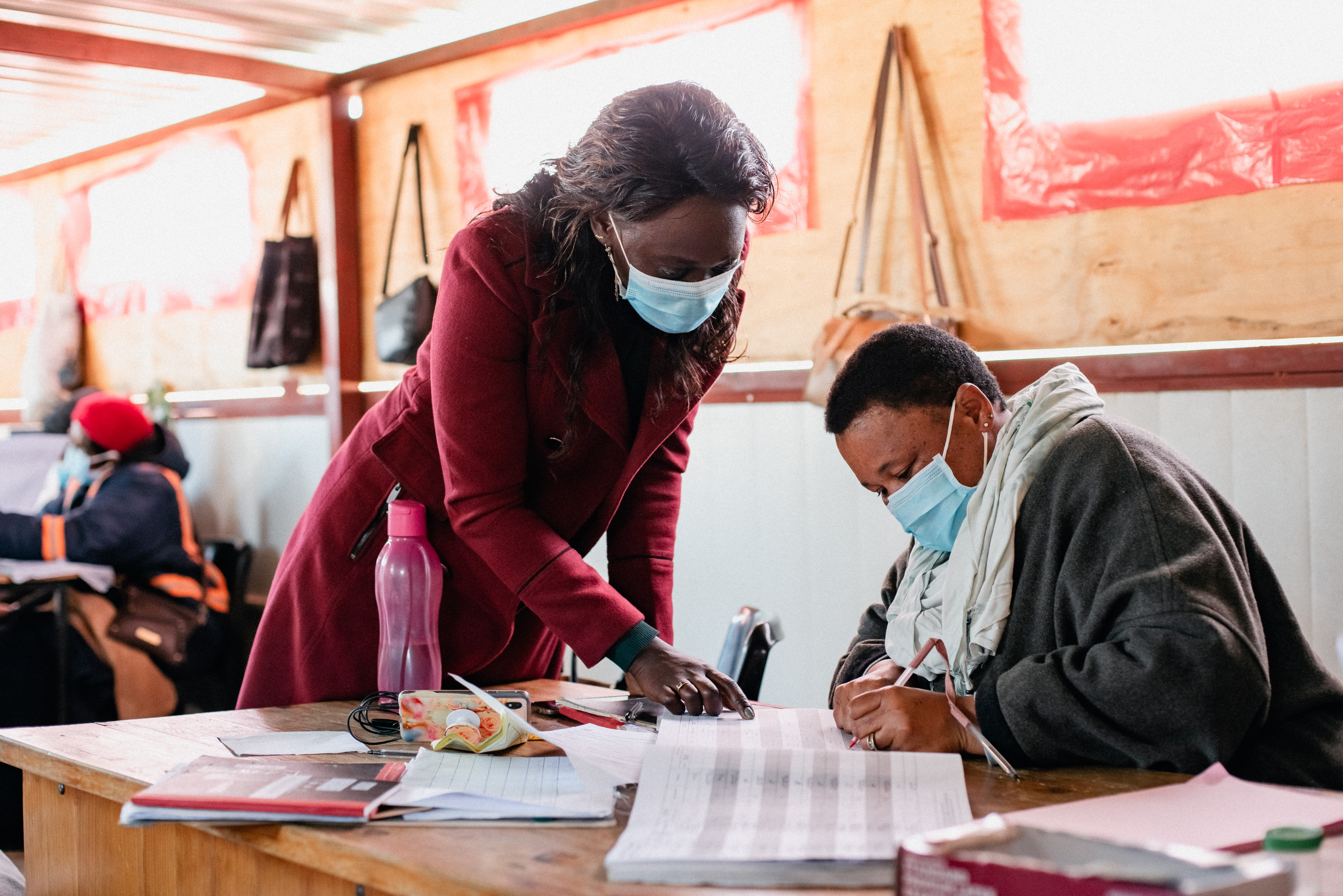}
    \includegraphics[height=65pt]{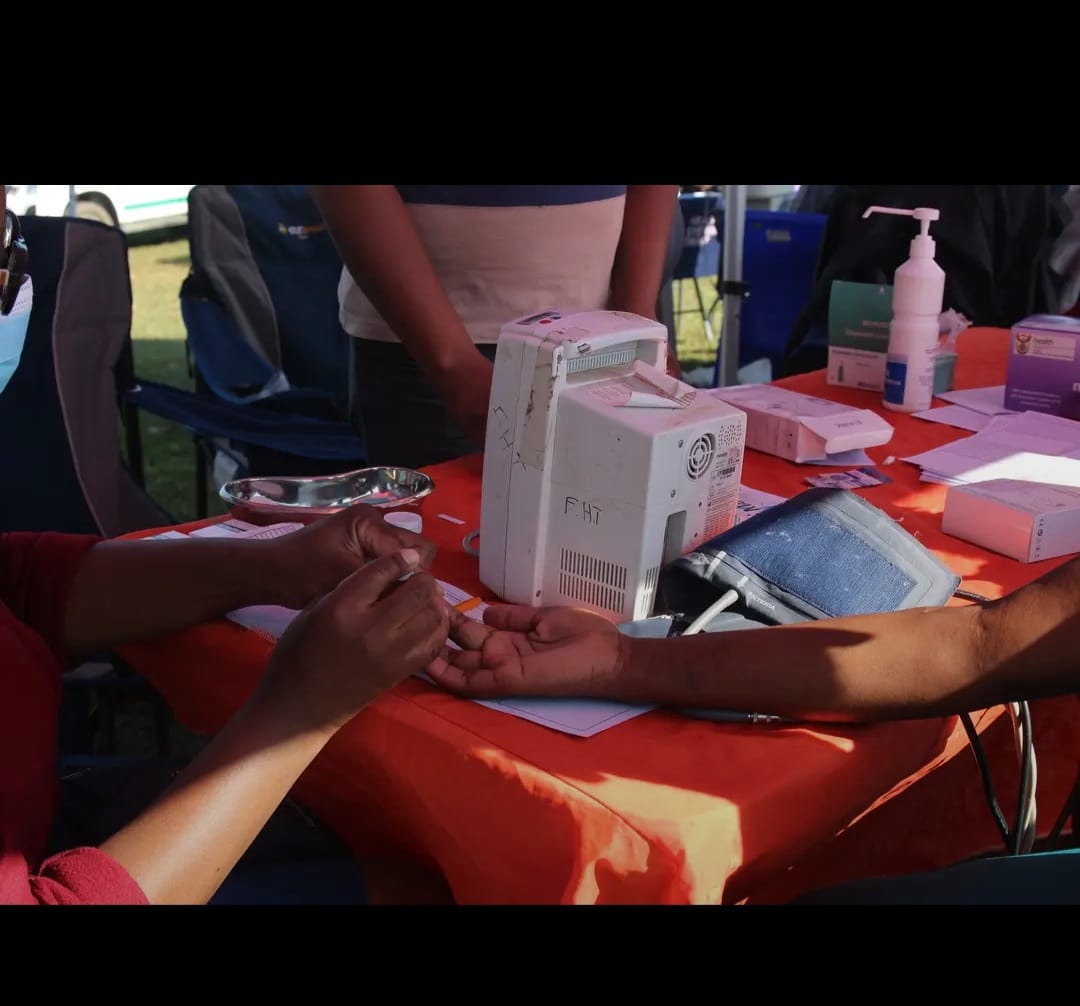}
    \includegraphics[height=65pt]{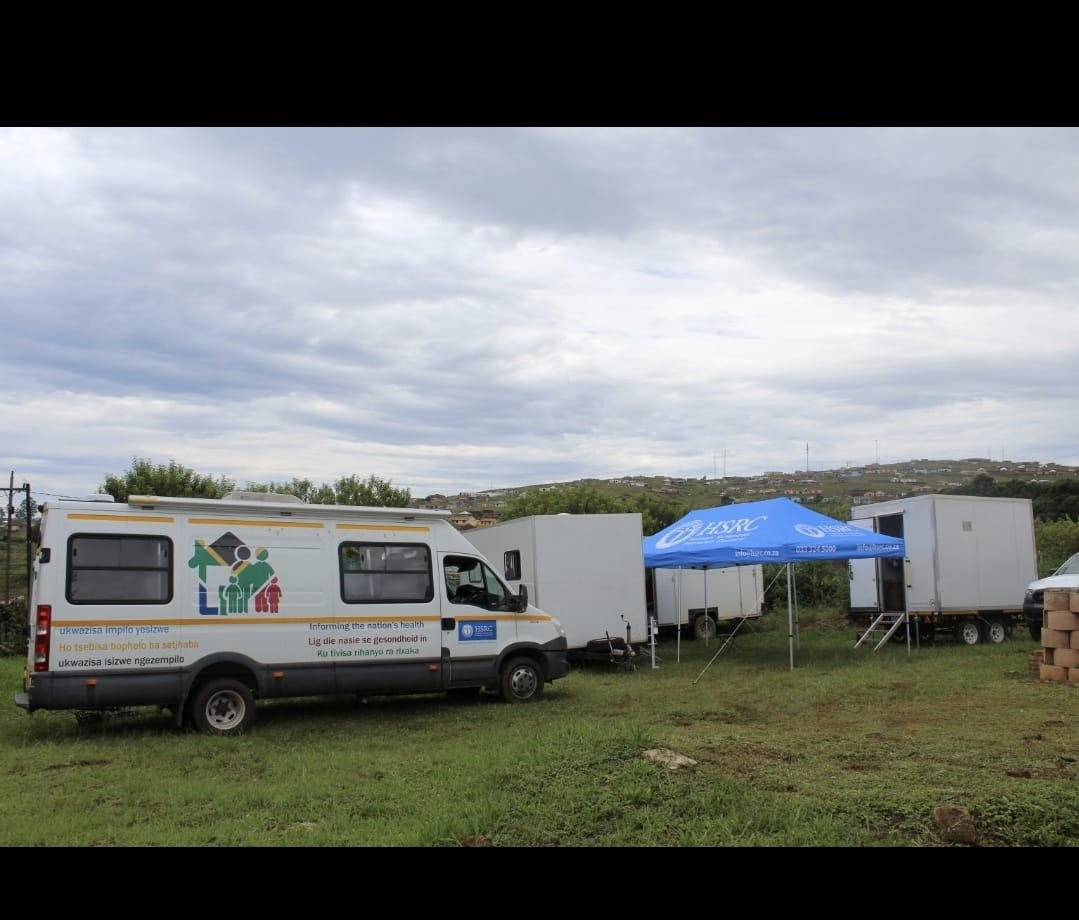}
    \caption{Images from our collaborators conducting HIV prevention field work in South Africa.}
    \label{fig:field-work}
\end{figure}

\textbf{Contributions.}
Our theory and simulations show that our proposed \texttt{CAST} algorithm has strong potential for real-world impact in improving the efficiency of HIV prevention resources, even under more realistic problem instances with imperfect networks.
More specifically:
\begin{enumerate}[leftmargin = *]
    \item \textbf{Intervention Model}:
    We introduce and formalize the Minimizing Infection via Node Treatment (\texttt{MINT}) problem where we have resources to ``treat'' $k$ out of a set $\mathbf{P}$ of virally unsuppressed individuals, a novel optimization problem that captures the objectives and constraints of the real-world limited resource allocation problem described above.
    We analyze and prove important mathematical properties of this problem, showing that its minimization objective is non-negative, monotonic, and submodular, and that \texttt{MINT} reduces to a special case of the more well-studied Influence Minimization (\texttt{IMIN}) problem but remains distinct from Influence Maximization (\texttt{IM}).
    \item \textbf{Cascade-Aware Suppression of Transmission (\texttt{CAST})}:
    For any fixed realization of the stochastic transmission network, we prove that \texttt{MINT} is equivalent to the Minimum-$k$-Union (\texttt{MkU}) problem.
    By combining \texttt{MkU} approximation algorithms with structured sample-based set constructions and Hoeffding-style concentration bounds, we provide \texttt{CAST}, a polynomial-time $(\delta, \epsilon)$-approximation algorithm for the general \texttt{MINT} problem with a ratio of $2\sqrt{|\mathbf{P}|}$ for expected new infections.
    
    \item \textbf{Empirical Evaluation}: We evaluate \texttt{CAST} on public real-world HIV network datasets in simulation, and compare its performance to several baseline algorithms from related public health and computer science literature. We show that \texttt{CAST} consistently outperforms these baselines across a variety of \texttt{MINT} instances and on disease networks beyond HIV.
    \item \textbf{Extension to Imperfect Networks}:
    To address practical deployment challenges regarding network collection and observability, we study the effectiveness of treatment decisions on incomplete networks that are less costly and quicker to record, which we supplement with learned link predictions. Since edge probabilities can also be difficult and expensive to estimate, we further evaluate how \texttt{CAST} adapts to mismatched edge probability initializations. Our empirical simulations show that \texttt{CAST} remains robust to these imperfect network settings.
\end{enumerate}


\section{Preliminaries and Related Work}

\paragraph{Notation.}
We use boldface letters (e.g., $\mathbf{S}$) for sets and calligraphic letters (e.g., $\mathcal{G}$, $\mathcal{H}$) for graphs and random realizations.

\textbf{Concentration bounds.}
We will use standard concentration inequalities to relate empirical estimates to expectations.
In particular, we use the following form of Hoeffding's inequality.

\begin{theorem}[Hoeffding's inequality \cite{hoeffding1963probability}]
\label{thm:hoeffding}
Let $X_1, \ldots, X_n$ be independent random variables such that $a_i \leq X_i \leq b_i$.
Then, for any $\eps > 0$,
\[
\Pr \left( \left| \sum_{i=1}^n X_i - \mathbb{E}[\sum_{i=1}^n X_i] \right| > \eps \right) \leq 2 \exp \left( - \frac{2 \eps^2}{\sum_{i=1}^n (b_i - a_i)^2} \right)
\]
\end{theorem}

\textbf{Probabilistic spread model.}
In one-shot intervention settings like ours, we are interested in the ultimate lifetime spread outcome rather than granular, step-by-step temporal dynamics.
We also need to consider spread probabilistically: for sexually transmitted infections like HIV, there may be network edges that never result in transmission \cite{dosekun2010overview}.
To this end, we employ the Independent Cascade (IC) model \cite{kempe2003maximizing} to characterize the diffusion of infections across the network $\mathcal{G}$.
In the IC model, when a node $v$ becomes active, it has a single discrete opportunity to activate each successor $u$ with probability $w_{(v,u)}$.
The type and frequency of such shared behavior between two individuals carry a specific risk, making the IC model's ``probabilistic coin flip'' a natural representation in our scenario in which an edge activation can be viewed as a cumulative lifetime chance of infection. 
Related models like Susceptible-Infected (SI) \cite{nakamura2017efficient} can capture more granular transient states and temporal dynamics, but do not offer a natural measurement of the lifetime spread outcome, which is exactly what we are ultimately interested in studying in this work.
In order to adapt such models to one-shot problems, we would need to either set an arbitrary horizon after which we no longer care about spread, or consider an infinite horizon, under which SI degenerates to simple reachability from source nodes.
From a mathematical perspective, the IC model has a \emph{live-edge} interpretation: each edge is independently retained with probability $w_{(v,u)}$ to yield a random graph realization $\mathcal{H} \sim \mathcal{G}$.
Under the live-edge perspective, the diffusion process is equivalent to reachability in $\mathcal{H}$, so expectations are calculated over the distribution of random graph realizations.


\textbf{Influence maximization.}
Influence maximization seeks a seed set that maximizes expected spread under stochastic influence models such as the Independent Cascade (IC) model and \emph{Linear Threshold} (LT) models \cite{kempe2003maximizing}.
In contrast to the IC model, the LT model captures cumulative influence, where a node activates once the total influence from its neighbors exceeds a threshold \cite{granovetter1978threshold,schelling2006micromotives}, which is less appropriate in our setting.
A key feature of both models is that the expected spread is monotone submodular, enabling a $(1 - 1/e)$-approximation via greedy algorithms \cite{kempe2003maximizing}.
Other influence models and corresponding algorithms have since been proposed and studied beyond IC and LT, including Triggering \cite{khim2016computing}, Competitive Influence \cite{bharathi2007competitive}, and Time-Delayed diffusion \cite{chen2012time}, though IC clearly remains the best choice for our setting.

\begin{definition}[Influence Maximization (\texttt{IM}) under IC]
\label{def:im}
Given a directed probabilistic graph $\mathcal{G} = (\mathbf{V}, \mathbf{E}, w)$ and a budget $k \in \mathbb{N}$, the \texttt{IM} problem is to choose a seed set $\mathbf{S} \subseteq \mathbf{V}$ with $|\mathbf{S}| \leq k$ to maximize $\mathbb{E}_{\mathcal{H} \sim \mathcal{G}}[ f_{\mathcal{H}}^{\mathrm{\texttt{IM}}}(\mathbf{S}) ]$ where
$
f_{\mathcal{H}}^{\mathrm{\texttt{IM}}}(\mathbf{S}) = \left| \{ u \in \mathbf{V} : \text{$u$ is reachable from $\mathbf{S}$ in $\mathcal{H}$} \} \right|
$.
\end{definition}

\textbf{Influence minimization and node blocking.}
Our work is most closely related to influence minimization via node blocking \cite{wang2024efficient,xie2025influence}.
The works in this literature typically fix the seed active set $\mathbf Q$ and focus on blocking a set of inactive nodes from $\mathbf V \setminus \mathbf Q$ to minimize total active nodes after influence diffuses through the network.
Under the live-edge interpretation, \texttt{IMIN} is equivalent to removing nodes from each realization $\mathcal{H}$ and maximizing the reduction in nodes reachable from $\mathbf{Q}$.
A key challenge in \texttt{IMIN} is that the objective is \emph{not submodular}, which precludes standard greedy-based approximations and necessitates heuristic or relaxation-based approaches \cite{wang2024efficient, wang2025time}.
We will later show that \texttt{MINT}'s objective can be reduced to a special case of \texttt{IMIN} that turns out to be an easier problem to solve: our objective is submodular and thus we may leverage known results in submodular \emph{minimization}.

\begin{definition}[Influence Minimization (\texttt{IMIN}) under IC]
\label{def:IBM}
Let $\mathcal{G} = (\mathbf{V}, \mathbf{E}, w)$ be a directed probabilistic graph with an initial set of active nodes $\mathbf{Q} \subseteq \mathbf{V}$.
Given a budget $k \in \mathbb{N}$, the \texttt{IMIN} problem is to choose $\mathbf{S} \subseteq \mathbf{V} \setminus \mathbf Q$ with $|\mathbf{S}| \leq k$ to minimize $\mathbb{E}_{\mathcal{H} \sim \mathcal{G}}\!\left[f_{\mathcal{H}}^{\mathrm{\texttt{IMIN}}}(\mathbf{S})\right]$ where
$
f_{\mathcal{H}}^{\mathrm{\texttt{IMIN}}}(\mathbf{S})
= \left| \left\{ u \in \mathbf{V}: \text{$u$ is reachable from $\mathbf{Q}$ in $\mathcal{H}[\mathbf{V} \setminus \mathbf{S}]$} \right\} \right|
$.
\end{definition}


\textbf{Minimum-$k$-Union (\texttt{MkU}) and submodular minimization.}
The Minimum-$k$-Union (\texttt{MkU}) problem is a central challenge in combinatorial optimization that is closely related to the Maximum Coverage problem, both of which are known to be NP-hard.
Both problems deal with submodular objectives, but Maximum Coverage admits a $(1 - 1/e)$ approximation (which is also known to be tight) \cite{feige1998threshold}, while \texttt{MkU} only has polynomial time algorithms achieving $\mathcal{O}(n^{\frac{1}{4} - \eps})$ and $2 \sqrt{m}$ approximations (which are roughly tight under certain complexity conjectures) \cite{chlamtavc2017minimizing,chlamtac2018densest}.

We formally define \texttt{MkU} in \cref{def:MkU} and will utilize the $2 \sqrt{m}$ approximation of \cite{chlamtac2018densest} as a subroutine in our work.
There are also extensions of this approximation algorithm that efficiently approximate submodular minimization problems in general \cite{chen2024extensions}.

\begin{definition}[\texttt{MkU}]
\label{def:MkU}
Consider $m$ sets $\mathbf{S}_1, \ldots, \mathbf{S}_m \subseteq \mathbf{U}$ over a universe $\mathbf{U}$ of $n$ elements.
Given $k \in \mathbb{N}$, the goal is to find a collection of $k$ subsets with the minimum union size.
\end{definition}

\textbf{Strategic network interdiction.}
Several network-based interventions for preventing spread of disease have previously been studied.
One particularly relevant area of research is on targeted distribution of pre-exposure prophylaxis (PrEP), which is used to protect individuals who do \emph{not} have HIV but are at risk of infection \cite{roberts2022impact}.
One common approach is to use learned models to predict individual risk scores to determine which individuals are at greatest risk of infection \cite{marcus2019use, haas2021machine}, rather than considering potential downstream network impacts.
Beyond HIV PrEP, prior work has studied general vaccine distribution to minimize expected new infections, both in one-shot settings \cite{chen2010better, cheng2020outbreak} and in sequential settings \cite{dong2024integrating, ling2024cooperating}.
In contrast to these works that focus on vaccination (i.e., ``removing negative nodes'', similar to the \texttt{IMIN} problem), our work focuses on virally suppressing infected individuals (i.e., ``removing positive nodes'') which admits a different mathematical structure.
Prior work in optimizing interventions on people living with HIV has been more related to casefinding through network-based testing rather than achieving viral suppression. Several studies have shown that targeting networks of recently infected individuals for contact tracing and testing can lead to better prevention outcomes \cite{nikolopoulos2016network, morgan2019network}.

\section{Formalizing Real-World Interventions}
\label{sec:mint}
We formalize our real-world problem setting from \cref{intro} in \cref{def:mint}. Some proofs throughout this section are deferred to \cref{sec:appendix-proofs}.

\begin{definition}[Minimizing Infection via Node Treatment (\texttt{MINT})]
\label{def:mint}
Let $\mathcal{G} = (\mathbf{V}, \mathbf{E}, w)$ be a directed weighted disease transmission graph.
Vertices represent individuals, partitioned into positive (infected) nodes $\mathbf{P} \subseteq \mathbf{V}$ and negative (uninfected) nodes $\mathbf{N} = \mathbf{V} \setminus \mathbf{P}$.
A directed edge $(v,u) \in \mathbf{E}$ indicates that infection may propagate from $v$ to $u$, and becomes \emph{live} independently with probability $w_{(v,u)}$.
Given a budget $k \in \mathbb{N}$, the \texttt{MINT} problem is to select a subset $\mathbf{S} \subseteq \mathbf{P}$ with $|\mathbf{S}| \leq k$ to minimize the expected number of nodes in $\mathbf{N}$ reachable via live edges from the remaining untreated sources $\mathbf{P} \setminus \mathbf{S}$.
\end{definition}

The existence of an edge $(v, u)$ implies that the corresponding individuals engage in transmission risk behavior such as sexual acts or needle sharing. The edge weight $w_{(v,u)}$ indicates the specific probability that $v$ will transmit disease to $u$ if they are or become infected, which depends on factors like risk behavior frequency, condom use, and viral load \cite{boily2009heterosexual, pearson2007modeling}. As positive nodes cannot infect other nodes that are already positive, we may assume without loss of generality that there are no incoming edges to nodes in $\mathbf{P}$ in the \texttt{MINT} formulation.
This assumption ensures that any directed path originating from a positive node has only negative internal vertices.

In the probabilistic transmission model of \texttt{MINT}, there are $2^{|\mathbf{E}|}$ possible live-edge realizations.
For any realization $\mathcal{H} = (\mathbf{V}, \mathbf{E}_{\mathcal{H}})$ with $\mathbf{E}_{\mathcal{H}} \subseteq \mathbf{E}$, we have
\[
\Pr(\mathcal{H})
=
\prod_{e \in \mathbf{E}}
w_e^{\mathbbm{1}[e \in \mathbf{E}_{\mathcal{H}}]}
\cdot
(1 - w_e)^{\mathbbm{1}[e \not\in \mathbf{E}_{\mathcal{H}}]}
\]

Given a subset $\mathbf{S} \subseteq \mathbf{P}$ of positive individuals to treat, we remove these nodes from the graph and measure the remaining exposure.
More formally, for a fixed realization $\mathcal{H}$, we define
\begin{equation}
\label{eq:realized-f}
f_{\mathcal{H}}(\mathbf{S})
= \left| \left\{ u \in \mathbf{N} : \text{$u$ is reachable from $\mathbf{P} \setminus \mathbf{S}$ in the directed graph $\mathcal{H}[\mathbf{V} \setminus \mathbf{S}]$} \right\} \right|
\end{equation}
That is, $f_{\mathcal{H}}(\mathbf{S})$ counts the number of negative nodes that remain reachable from untreated positives \emph{after} removing treated nodes.
Taking expectation over realizations, we define
\begin{equation}
\label{eq:F}
F(\mathbf{S})
= \mathbb{E}_{\mathcal{H} \sim \mathcal{G}} \left[ f_{\mathcal{H}}(\mathbf{S}) \right]
\end{equation}
That is, the \texttt{MINT} objective is exactly minimizing $F(\mathbf{S})$ under the constraint of $\mathbf{S} \subseteq \mathbf{P}$ and $|\mathbf{S}| \leq k$.


With the objective represented in terms of $F$, \cref{prop:special-case-of-ibm} tells us that the \texttt{MINT} problem reduces to a special case of \texttt{IMIN}.
In fact, \texttt{MINT} admits a different mathematical structure that is easier to tackle: \cref{prop:F-is-submodular} shows that the objective is submodular, so \texttt{MINT} is a special case of submodular minimization and we may leverage on known results in submodular minimization.

\begin{proposition}
\label{prop:special-case-of-ibm}
\texttt{MINT} can be reduced to a special case of \texttt{IMIN} (\cref{def:IBM}).
\end{proposition}
\begin{proof}
Let $\mathcal{G} = (\mathbf{V}, \mathbf{E}, w)$ be a \texttt{MINT} instance with $\mathbf{P} \subseteq \mathbf{V}$.
Initialize an auxiliary \texttt{IMIN} graph $\mathcal{G}^\prime = (\mathbf{V}^\prime, \mathbf{E}^\prime, w^\prime)$ with $\mathbf{E}^\prime = \mathbf{E}$, $w^\prime = w$ and $\mathbf{V}^\prime = \mathbf{V} \cup \mathbf{Q}$, where $|\mathbf{Q}| = |\mathbf{P}|$.
In $\mathcal{G}^\prime$, nodes from $\mathbf{V}$ are inactive and the initial active nodes are $\mathbf{Q}$.
We associate each $p \in \mathbf{P}$ with a unique active node $q \in \mathbf{Q}$, then add a directed edge $(q,p)$ to $\mathbf{E}^\prime$ with weight $w^\prime_{(q, p)} = 1$.
In this construction, minimizing the \texttt{MINT} objective $F(\mathbf{S})$ on $\mathcal{G}$ is equivalent to minimizing the \texttt{IMIN} objective on $\mathcal{G}^\prime$, subject to the constraint of $\mathbf{S} \subseteq \mathbf{P} \subseteq (\mathbf{V}^\prime \setminus \mathbf{Q})$, rather than the more general case where $\mathbf{S} \subseteq (\mathbf{V}^\prime \setminus \mathbf{Q})$.
\end{proof}

\begin{restatable}{proposition}{Fissubmodular}
\label{prop:F-is-submodular}
$F(\mathbf{S})$ is non-negative, monotonically non-increasing, and submodular.    
\end{restatable}

As both \texttt{MINT} and \texttt{IM} involve selecting a subset of individuals to optimize stochastic influence on the network via independent cascade, it is tempting to view \texttt{MINT} as a variant of influence maximization.
However, the two problems are fundamentally different in their combinatorial structure.
\texttt{IM} is a coverage problem: it seeks a set $\mathbf{S}$ that maximizes the size of a union of reachable sets.
In contrast, \texttt{MINT} is an \emph{interdiction} problem: it seeks a set $\mathbf{S}$ whose \emph{removal} minimizes the size of the union of reachable sets originating from the \emph{remaining} unsuppressed nodes $\mathbf{P} \setminus \mathbf{S}$.
These objectives are not equivalent under any complement-based transformation of the seed set, primarily because \texttt{IM} focuses on what is gained through selection, while \texttt{MINT} focuses on what remains after suppression. 
\cref{prop:distinct} formalizes this distinction via a concrete counterexample. Nevertheless, we include this complement-based \texttt{IM} approach as a baseline in our experiments.

\begin{proposition}
\label{prop:distinct}
The \texttt{MINT} problem is distinct from \texttt{IM}. Specifically, let $\mathbf{S}^*_{\texttt{MINT}}$ be an optimal solution for \texttt{MINT} and $\mathbf{S}^*_{\texttt{IM}}$ be an optimal solution for \texttt{IM} restricted to the same graph $\mathcal{G}$, candidate set $\mathbf{P}$, and budget $k$.
There exist graph topologies $\mathcal{G}$ such that $\mathbf{S}^*_{\texttt{IM}}$ is strictly suboptimal for the \texttt{MINT} objective.
\end{proposition}
\begin{proof}
Consider the graph in \cref{fig:different-from-influence-maximization}, and suppose $w_e = 1$ for all edges so that there is only one live-edge realization.
Suppose $\mathbf{P} = \{a,b,c,d\}$.
With a budget of $k = 2$, we see that $\mathbf{S}^\star_{\texttt{MINT}} = \{c,d\}$ and $a \in \mathbf{S}^\star_{\texttt{IM}}$, resulting in $F(\mathbf{S}^\star_{\texttt{MINT}}) = 1$ and $F(\mathbf{S}^\star_{\texttt{IM}}) \geq 2$.
\end{proof}

\begin{figure}[htb]
\centering
\input{different-from-influence-maximization}
\caption{Suppose $\mathbf{P} = \{a,b,c,d\}$. With a budget of $k = 2$, we should treat $\mathbf{S} = \{c,d\}$ resulting in $F(\mathbf{S}) = 1$. Meanwhile, an influence maximization approach will pick $\mathbf{S'}$ with node $a \in \mathbf{S'}$ to ensure that all nodes are influenced. However, doing so necessarily causes $F(\mathbf{S}') \geq 2$.}
\label{fig:different-from-influence-maximization}
\end{figure}
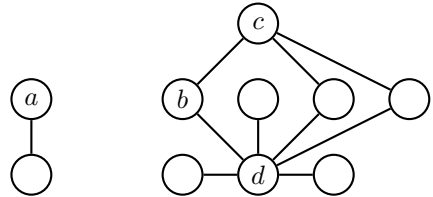

\section{Cascade-Aware Suppression of Transmission (\texttt{CAST})}
\label{sec:fixed-realization}

We now introduce and analyze our approach to \texttt{MINT}, which we call Cascade-Aware Suppression of Transmission (\texttt{CAST}); see \cref{alg:known}.
We begin by showing a useful relationship between \texttt{MINT} and the \texttt{MkU} problem that we will exploit in the design of \texttt{CAST}.


\begin{algorithm}[htb]
\caption{Cascade-Aware Suppression of Transmission (\texttt{CAST})}
\label{alg:known}
\begin{algorithmic}[1]
\Require Infection network $\mathcal{G} = (\mathbf{V} = \mathbf{P} \cup \mathbf{N}, \mathbf{E}, w)$, budget $k \in \mathbb{N}$, error tolerance  $\epsilon$, probability $\delta$
\Ensure $\mathbf{S} \subseteq \mathbf{P}$ \Comment{$|\mathbf{S}| \leq k$}
    \State Define $m = \left\lceil \frac{|\mathbf{N}|^2}{2\eps^2} \ln \frac{2^{|\mathbf{P}| + 1}}{\delta} \right\rceil$ \Comment{See \cref{lem:algo-analysis}}
    \State Sample $m$ realizations $\mathcal{H}_1, \ldots, \mathcal{H}_m \sim \mathcal{G}$ \Comment{Possibly repeated}
    \State For each $v \in \mathbf{P}$, define reachability sets
    $
    \mathbf{R}_v = \left\{ (i, u) \in [m] \times 
    \mathbf{N}: \text{$u$ is reachable from $v$ in $\mathcal{H}_i$} \right\}
    $
    \State $\mathbf{S}^\prime \gets$ Run the \texttt{MkU} algorithm of \cite{chlamtac2018densest} on the sets $\{\mathbf{R}_v\}_{v \in \mathbf{P}}$ with parameter $|\mathbf{P}| - k \in \mathbb{N}$
    \State \Return $\mathbf{S} = \mathbf P \setminus \mathbf{S}^\prime$
\end{algorithmic}
\end{algorithm}



\begin{proposition}
\label{prop:MkU-equivalence}
For any fixed realization $\mathcal{H} \sim \mathcal{G}$, the problem of finding a subset $\mathbf{S} \subseteq \mathbf{P}$ with $|\mathbf{S}| \leq k$ that minimizes $f_{\mathcal{H}}(\mathbf{S})$ is equivalent to the \texttt{MkU} problem with parameter $k' = |\mathbf{P}| - k$.
\end{proposition}
\begin{proof}
Recall the problem definitions of \texttt{MkU} and \texttt{MINT} in \cref{def:MkU} and \cref{def:mint} respectively.
We demonstrate this equivalence by providing a bidirectional reduction between the two problems.

\textbf{(\texttt{MINT} $\to$ \texttt{MkU}):}
Given a realization $\mathcal{H}$, let the \texttt{MkU} universe be $\mathbf{U} = \mathbf{N}$.
For each positive node $v \in \mathbf{P}$, define a set $\mathbf{R}_v \subseteq \mathbf{U}$ containing all negative nodes reachable from $v$ in $\mathcal{H}$.
Let $\mathbf{S} \subseteq \mathbf{P}$ be the set of treated nodes with $|\mathbf{S}| \leq k$.
The set of untreated nodes is $\mathbf{P}' = \mathbf{P} \setminus \mathbf{S}$, where $|\mathbf{P}'| \geq |\mathbf{P}| - k$.
Because nodes in $\mathbf{P}$ have no incoming edges, a node $u \in \mathbf{N}$ is reachable from the set $\mathbf{P} \setminus \mathbf{S}$ if and only if it is reachable from at least one node $v \in \mathbf{P} \setminus \mathbf{S}$ in the original realization $\mathcal{H}$.
Thus, $f_{\mathcal{H}}(\mathbf{S}) = |\bigcup_{v \in \mathbf{P} \setminus \mathbf{S}} \mathbf{R}_v|$.
Minimizing $f_{\mathcal{H}}(\mathbf{S})$ is therefore equivalent to selecting $k' = |\mathbf{P}| - k$ sets from $\{\mathbf{R}_v\}_{v \in \mathbf{P}}$ to minimize the size of their union.

\textbf{(\texttt{MkU} $\to$ \texttt{MINT}):}
Consider an \texttt{MkU} instance with sets $\mathcal{X} = \{X_1, \ldots, X_m\}$ over universe $\mathbf{U} = \{u_1, \ldots, u_n\}$ and a budget $k'$.
We construct a directed bipartite graph $\mathcal{H} = (\mathbf{V}, \mathbf{E}_{\mathcal{H}})$ where $\mathbf{P} = \{p_1, \dots, p_m\}$ and $\mathbf{N} = \{n_1, \dots, n_n\}$.
We add a directed edge $(p_i, n_j)$ if and only if $u_j \in X_i$.
In this construction, $p_i$ has no incoming edges.
Selecting $k'$ sets to minimize their union in the \texttt{MkU} instance is identical to selecting $k'$ untreated nodes in \texttt{MINT} (where the treatment budget is $k = m - k'$) to minimize the number of reachable negative nodes.
\end{proof}

This equivalence implies that \texttt{MINT} is NP-hard and inherits the inapproximability results of \texttt{MkU}.
Crucially, it justifies our use of the $2 \sqrt{|\mathbf{P}|}$ approximation algorithm for the set-union objective in \cref{alg:known}.
Furthermore, the proof highlights the importance of the ``no incoming edges for $\mathbf{P}$'' property: without it, the reachability of a node $u$ could depend on paths through other treated nodes, breaking the set-union structure.

Leveraging the structural equivalence from \cref{prop:MkU-equivalence}, we present our main theoretical result for this section, which is that \texttt{CAST} (\cref{alg:known}) is a polynomial time approximation algorithm.

\begin{restatable}{theorem}{approximationratio}
\label{thm:approximation_ratio}
For any fixed $\eps, \delta \in (0,1)$, let $\mathbf{S}^\star$ be the true optimizer for $F$ and $\hat{\mathbf{S}}$ be the set returned by \texttt{CAST} (\cref{alg:known}).
Then, we have $\Pr \left[ F(\hat{\mathbf{S}}) \leq 2 \sqrt{|\mathbf{P}|} \cdot F(\mathbf{S}^\star) + \eps \cdot (1 + 2 \sqrt{|\mathbf{P}|}) \right] \geq 1 - 2 \delta$.
Furthermore, \texttt{CAST} runs in polynomial time with respect to $|\mathbf{V}|$, $|\mathbf{E}|$, $\eps^{-1}$, and $\delta^{-1}$.
\end{restatable}

To establish \cref{thm:approximation_ratio}, we prove the following two lemmas.
\cref{lem:set_construction} examines the structure of each set $\mathbf{R}_v$ in \texttt{CAST} (\cref{alg:known}) and reveals how these set constructions bridge the gap between the single realization case from \cref{prop:MkU-equivalence} to the general empirical objective over many sampled realizations.
Meanwhile, \cref{lem:algo-analysis} demonstrates the importance of the choice of $m$, relating the empirical sample mean $\hat{F}_m$ to the true expectation $F$ by applying Hoeffding's inequality (\cref{thm:hoeffding}) to the random variables $X_i = f_{\mathcal{H}_i}(\mathbf{S})$.

\begin{restatable}{lemma}{setconstruction}
\label{lem:set_construction}
For any $m$ realizations $\mathcal{H}_1, \ldots, \mathcal{H}_m \sim \mathcal{G}$ and $\mathbf{S} \subseteq \mathbf{P}$, define $\hat{F}_m(\mathbf{S}) = \frac{1}{m} \sum_{i=1}^m f_{\mathcal{H}_i}(\mathbf{S})$ as the empirical version of $F(\mathbf{S})$ defined on the $m$ realizations.
Then, solving \texttt{MkU} on $\{\mathbf{R}_v\}_{v \in \mathbf{P}}$ and $k' = |\mathbf{P}| - k \in \mathbb{N}$ in \texttt{CAST} (\cref{alg:known}) produces $\mathbf{S}^\prime$ such that $\mathbf{S}  = \mathbf{P} \setminus \mathbf{S}^\prime$ optimally minimizes $\hat{F}_m(\mathbf{S})$, with $|\mathbf{S}| \leq k$.
\end{restatable}
\begin{proof}
By definition, $\hat{F}_m(\mathbf{S}) = \frac{1}{m} \sum_{i=1}^m | \{ u \in \mathbf{N} : \exists v \in \mathbf{P} \setminus \mathbf{S} \text{ s.t. } v \to u \text{ in } \mathcal{H}_i \} |$.
Let the universe be $\mathbf{U} = \{1, \dots, m\} \times \mathbf{N}$.
For each $v \in \mathbf{P}$, let $\mathbf{R}_v$ be the subset of $\mathbf{U}$ containing pairs $(i, u)$ such that $u$ is reachable from $v$ in realization $\mathcal{H}_i$, as in \cref{alg:known}.
Then,
\[
\left| \bigcup_{v \in \mathbf{P} \setminus \mathbf{S}} \mathbf{R}_v \right|
= \sum_{i=1}^m | \{ u \in \mathbf{N} : \exists v \in \mathbf{P} \setminus \mathbf{S} \text{ s.t. } v \to u \text{ in } \mathcal{H}_i \} | = m \cdot \hat{F}_m(\mathbf{S})
\]
So, minimizing the union of $|\mathbf{P}| - k$ sets in the universe $\mathbf{U}$ is precisely the \texttt{MkU} problem, and the claim follows as minimizing $m \cdot \hat{F}_m(\mathbf{S})$ is equivalent to minimizing $\hat{F}_m(\mathbf{S})$.
\end{proof}

\begin{restatable}{lemma}{algoanalysis}
\label{lem:algo-analysis}
Fix any $\eps > 0$ and $\delta \in (0,1)$.
With $m = \left\lceil \frac{|\mathbf{N}|^2}{2\eps^2} \ln \frac{2^{|\mathbf{P}| + 1}}{\delta} \right\rceil$, for any $\mathbf{S} \subseteq \mathbf{P}$, we have
\[
\Pr \left[ \left| F(\mathbf{S}) - \hat{F}_m(\mathbf{S}) \right| > \eps \right] \leq \delta
\]
\end{restatable}

\section{Experiments}
\label{sec:experiments}

We provide empirical evaluations of \texttt{CAST} compared to baseline methods through retrospective simulation studies on real-world infectious disease datasets.
Key experimental settings are described in the main text and \cref{sec:parameters}.
Across all runs in our experiments, the longest runtime for any algorithm was $\sim$30 minutes on an M4 Pro MacBook with 24 GB of memory, well within practical time limits for real use cases.
Our code will be published in a public GitHub repository in the final version of the paper.

\textbf{Datasets.}
We use publicly available anonymized real-world network data from ICPSR \cite{ICPSR} for our experiments. This dataset contains networks for several sexually transmitted infections including HIV, Syphilis, and Hepatitis. In addition to network connections and disease status, each node in a network is also associated with a covariate vector.
Additional details are provided in \cref{sec:parameters}.

\textbf{Baselines.}
We consider the following baselines in our experiments:

\begin{enumerate}[leftmargin=*]
    \item \textbf{No Blocking}: 
    Baseline for how many new infections occur without any intervention.
    \item \textbf{Random}: 
    Randomly selects a blocker set $\mathbf S$ with the maximal size allowed by the budget. 
    \item \textbf{Highest Degree}: This algorithm selects blocker nodes from $\mathbf{P}$ with the greatest unweighted degrees, i.e., individuals with the greatest number of contacts. This is a popular heuristic that has been previously studied for interventions on infectious disease networks \cite{liu2021efficient, brazia2026reconstructing}. 
    \item \textbf{SandIMIN}: We adapt the SandIMIN influence minimization algorithm \cite{wang2024efficient} with Monte-Carlo sampling to the \texttt{MINT} problem. We use weighted degree as the lightweight heuristic.
    \item \textbf{Greedy}: This algorithm greedily constructs $\mathbf{S}$ by adding nodes one by one with greatest marginal reduction in expected spread, which is estimated through Monte-Carlo realization samples. We include this baseline to evaluate naive minimization of the \texttt{MINT} objective.
    \item \textbf{IM}: This algorithm adapts the greedy influence maximization algorithm to the \texttt{MINT} problem. It first assumes that all nodes in $\mathbf{P}$ are blocked, then adds nodes one by one to $\mathbf{S}$ with greatest marginal increase in spread. This baseline demonstrates the misalignment between the objectives in \texttt{MINT} vs. \texttt{IM}, as discussed in \cref{prop:distinct}.
\end{enumerate}

\textbf{Known networks.}
For these experiments, we follow the steps below:
\begin{enumerate}[leftmargin=*]
    \item We convert a network from the ICPSR dataset into a directed transmission graph. All edges from the dataset are converted to bidirectional edges, then any edge going into a node in $\mathbf P$ is removed. This is because no node can transmit HIV to nodes in $\mathbf P$, since nodes in $\mathbf P$ are already individuals with HIV (as discussed in \cref{sec:mint}).
    \item We generate $a = 10$ weighted graphs $\mathcal{G}_1, \mathcal{G}_2, ..., \mathcal{G}_a$. In each graph, all edge probabilities are initialized independently by sampling an edge probability distribution $\mathcal{P}$.
    \item For each $i \in \{1, ..., a\}$, algorithms are given $\mathcal{G}_i$, intervention budget $k$, and $b = 50$ Monte-Carlo realizations for each single run on a single graph, and then they output a blocker set $\mathbf S_i$. 
    \item For each run on $\mathcal{G}_i$, average new infections are measured via IC with $c = 50$ Monte-Carlo realizations from $\mathcal{G}_i$. 
    We then report the mean $\pm$ SE of average new infections across the $a$ graphs.
    
\end{enumerate}

\cref{fig:core-result} shows performance of all algorithms on the HIV dataset under uniform random edge probabilities and varying intervention budgets.
Although our target application is in HIV prevention, we also provide experimental results on the datasets for Syphilis and Hepatitis, for which prevention efforts also face major resource constraints.
For the experiment on the HIV network, we make three main observations.
First, random interventions perform substantially worse than all other baselines, underscoring the importance of targeted, network-aware strategies.
Second, \texttt{CAST} consistently matches or outperforms all baselines, with the largest gains at low intervention budget regimes. 
For example, when intervening on only 20\% of all HIV-positive individuals, \texttt{CAST} incurs 56.1\% less average new infections than the next best baseline (Greedy) and 88.3\% less than No Blocking.
Third, \texttt{CAST} remains robust across budgets and exhibits low standard error across edge probability initializations.
\texttt{CAST} also maintains the strongest comparative advantage in other disease networks, especially for lower intervention budgets. 
For example, \texttt{CAST} incurs 45.9\% and 41.2\% less average new infections than the best baseline at the 20\% budget for the Syphilis and Hepatitis networks, respectively.

\begin{figure}
    \centering
    \includegraphics[width=\linewidth]{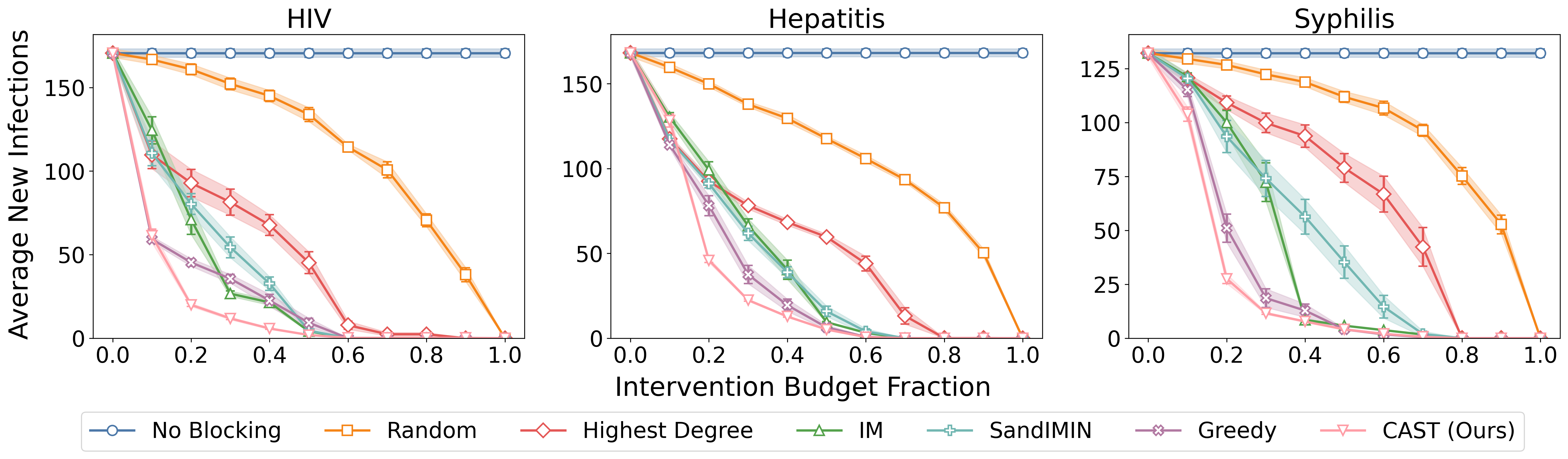}
    \caption{Average new infections ($\pm$ SE) for the HIV (left), Hepatitis (middle), and Syphilis (right) networks, under uniform random edge probabilities and varying intervention budgets.
    }
    \label{fig:core-result}
\end{figure}



\textbf{Ablations on the edge transmission probabilities.}
To study the robustness of our empirical evaluations, we consider alternative edge probability initializations on the HIV network, at the fixed intervention budget of $0.2$; see the first two plots in \cref{fig:ablations}.
We consider
(i) uniform edge probabilities, where $w_{(v,u)} = p$ for all edges with $p$ varying from $0.2$ to $0.8$, and
(ii) Gaussian edge probabilities, where $w_{(v,u)} \sim \mathcal{N}(\mu, \sigma^2)$ clipped to [0, 1] with $\sigma = 0.2$ and $\mu$ varying from $0.2$ to $0.8$.
As expected, the average number of new infections increases with edge probabilities, reflecting higher transmission likelihood.
Nevertheless, \texttt{CAST} remains consistently robust and substantially outperforms all baselines, with its advantage becoming more pronounced at higher transmission rates.
For instance, at $p=\mu=0.5$, \texttt{CAST} achieves $59.7\%$ (uniform) and $60.2\%$ (Gaussian) fewer average new infections compared to the next best baseline.
These gains improve to $81.3\%$ and $82.4\%$ respectively at $p=\mu=0.8$.


\begin{figure}
    \centering
    \includegraphics[width=\linewidth]{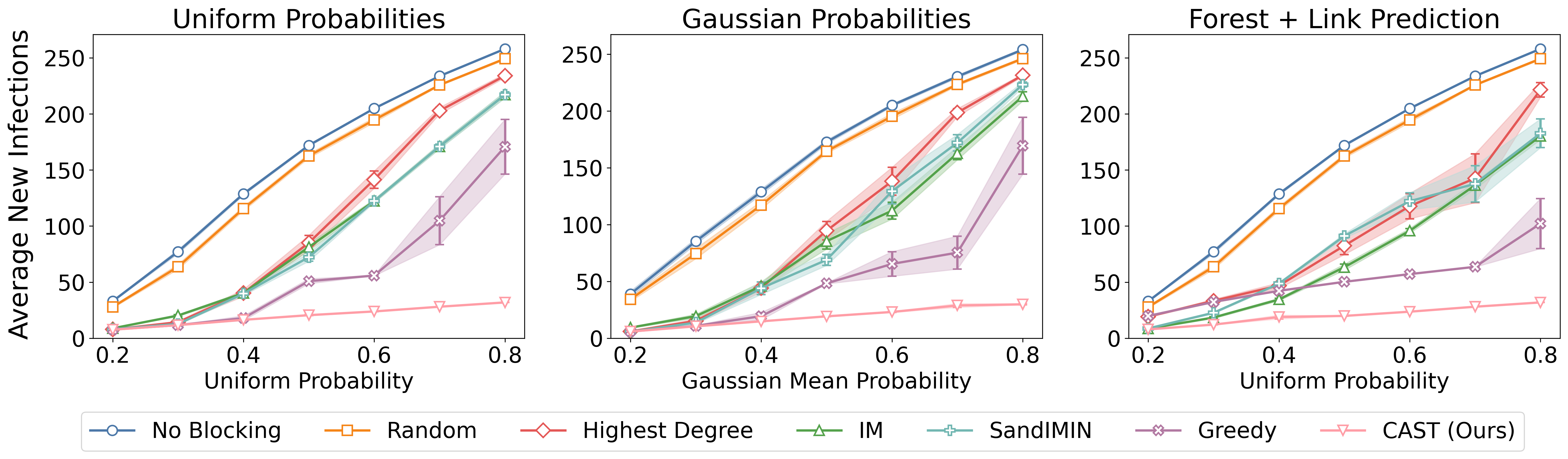}
    \caption{Average new infections ($\pm$ SE) on the HIV network with intervention budget fixed at 0.2 and varying edge probability initializations. We evaluate different uniform edge probabilities (left) and Gaussian edge initializations with different means (middle), with fixed $\sigma^2 = 0.2^2$. We also study performance across different uniform edge probabilities when only the forest subgraph of the transmission network is known, which we supplement with link prediction (right).}
    \label{fig:ablations}
\end{figure}


\textbf{Imperfect networks.} One practical challenge that network-based strategies face in general is in collecting network data. Thus far, our experiments assumed full observability of both the network structure and edge probabilities. In the following, we challenge these assumptions through evaluations focused on more practical problem instances. We use $\mathcal{G}^\prime$ to denote the decision graph that algorithms use to make decisions and $\mathcal{G}$ to denote the true graph that is used to evaluate those decisions. 

First, we consider imperfect network structures. In order to observe the full network structure, a public health official must expend significant time and resources on contact tracing, which may be prohibitively expensive. One option that can be more feasible in populations with high ongoing HIV risk is a peer-to-peer referral system known as Respondent-Driven Sampling (RDS). In these systems, individuals will recursively refer many of their relevant contacts to HIV services, allowing edges to be collected through referrals. However, individuals are unlikely to or may not be allowed to accept multiple referrals to the same service from different individuals. As observed in prior work \cite{stein2018stochastic, kangaslahti2026policy}, the resulting recorded network has a forest structure and the remaining edges are never recorded. Although some studies have recognized this observability gap and tried to predict hidden parts of the network structure \cite{chen2016seeing, crawford2018hidden}, they fall short of evaluating how well these predicted networks support downstream decision-making. Thus, it remains largely unclear how useful learned completions of forest structures can practically be in downstream applications such as network-based interventions.

To address this, we study a predict-then-optimize framework \cite{elmachtoub2022smart} for solving \texttt{MINT} given only the forest subgraph of the transmission network. Though we lose the performance guarantees of \texttt{CAST} under this model, we empirically evaluate how well it adapts compared to other baselines. Specifically, our framework operates as follows:
\begin{enumerate}[leftmargin=*]
    \item We perform a breadth-first search projection on each component of the HIV transmission network $\mathcal G$ to form $\mathcal F$, the referral graph with forest structure. $\mathcal F$ represents the subgraph of $\mathcal{G}$ that we could cheaply collect through peer-to-peer referrals.
    \item Using the complete networks from the other ICPSR disease network datasets, we train a probabilistic link prediction model $\ell(v, u)$ that predicts the probability of the existence of an edge between any node pair $(v, u)$ within a shared connected component.
    \item We use $\ell$ to infer edge probabilities between all node pairs within $\mathcal F$ that share the same connected component and do not already have an existing edge in $\mathcal F$. We then normalize these predicted probabilities to the scale of the existing edge probabilities in $\mathcal F$ and add the resulting weighted edges to the edge set of $\mathcal F$, forming the decision graph $\mathcal{G}^\prime$. 
    \item We task blocking algorithms to solve \texttt{MINT} on $\mathcal{G}^\prime$ and evaluate decisions on the true $\mathcal{G}$.
\end{enumerate}

We use an SVM-based link prediction model~\cite{scikit-learn}; implementation details are provided in \cref{sec:parameters}.
The rightmost plot in \cref{fig:ablations} reports results under this setting.
We evaluate using uniform edge probabilities for ease of normalization, though other distributions can be incorporated by scaling predicted probabilities with samples from the edge probability distribution.
Compared to the setting with full network knowledge (first plot in \cref{fig:ablations}), all methods --- including \texttt{CAST} --- exhibit only modest performance losses on graphs with low uniform edge probabilities.
For example, \texttt{CAST} incurs 4.08\%, 3.90\%, and 15.4\% more average new infections under uniform edge probabilities of 0.2, 0.3, and 0.4 (respectively).
At edge uniform probabilities of 0.6 and higher, \texttt{CAST} has roughly equal performance compared to the case with full network knowledge.
This suggests that a forest backbone obtained via peer referrals, augmented with link prediction, can serve as a cost-effective substitute for complete network reconstruction.
Since accurately estimating transmission risk on every edge can be expensive, we also study robustness to imperfect edge probabilities.
More specifically, we consider two forms of mismatch between true and assumed edge probabilities and find that \texttt{CAST} remains robust in both cases.
Details are deferred to \cref{sec:edge-probs}.



\section{Conclusion and Discussion}

Inspired by the real-world challenges faced by our public health collaborators in achieving and maintaining HIV epidemic control under growing resource constraints, we study the problem of optimally allocating limited interventions --- such as adherence counseling and high-efficacy treatment --- to minimize transmission cascades.
In collaboration with these domain experts, we formalize this challenge as the \texttt{MINT} problem, bridging real-world epidemiological objectives with constrained optimization.
We then propose \texttt{CAST} (\cref{alg:known}), a polynomial-time $(\delta, \epsilon)$-approximation algorithm with a $2\sqrt{|\mathbf{P}|}$ approximation guarantee.
Empirical results on real-world networks show that \texttt{CAST} consistently outperforms established baselines and remains robust to other resource-constrained diseases, imperfect network information, and varying transmission dynamics.



While we are primarily motivated by HIV treatment, the \texttt{MINT} framework applies broadly to other diseases and also beyond public health applications.
In social networks, \texttt{MINT} could model the deplatforming or suspension of a limited set of influencers ($\mathbf{P}$) to protect the susceptible user base ($\mathbf{N} = \mathbf{V} \setminus \mathbf{P}$) from the diffusion of harmful atrocity speech, e.g., \cite{jhaver2021deplatforming}.
In cybersecurity settings, \texttt{MINT} could model the removal of a subset of infected machines ($\mathbf{P}$) to minimize the expected infection of the broader network ($\mathbf{N} = \mathbf{V} \setminus \mathbf{P}$), e.g., \cite{karakose2018node}.

Discussion of broader impacts, ethics, fairness, and limitations to be addressed in future work are provided in \cref{sec:discussion}.

\bibliography{refs}
\bibliographystyle{alpha}


\appendix
\input{appendix-proofs}
\input{appendix-experiments}
\input{appendix-discussion}



\end{document}

%% file: different-from-influence-maximization.tex
\begin{tikzpicture}
\node[draw, thick, circle, minimum size=15pt, inner sep=0pt] at (0,0) (a) {$a$};
\node[draw, thick, circle, minimum size=15pt, inner sep=0pt] at (0,-1) (e) {};
\node[draw, thick, circle, minimum size=15pt, inner sep=0pt] at (2,0) (b) {$b$};
\node[draw, thick, circle, minimum size=15pt, inner sep=0pt] at (3,1) (c) {$c$};
\node[draw, thick, circle, minimum size=15pt, inner sep=0pt] at (3,-1) (d) {$d$};
\node[draw, thick, circle, minimum size=15pt, inner sep=0pt] at (4,0) (f) {};
\node[draw, thick, circle, minimum size=15pt, inner sep=0pt] at (5,0) (g) {};
\node[draw, thick, circle, minimum size=15pt, inner sep=0pt] at (2,-1) (h) {};
\node[draw, thick, circle, minimum size=15pt, inner sep=0pt] at (3,0) (i) {};
\node[draw, thick, circle, minimum size=15pt, inner sep=0pt] at (4,-1) (j) {};

\draw[thick] (a) -- (e);
\draw[thick] (b) -- (c);
\draw[thick] (b) -- (d);
\draw[thick] (c) -- (f);
\draw[thick] (c) -- (g);
\draw[thick] (d) -- (f);
\draw[thick] (d) -- (g);
\draw[thick] (d) -- (h);
\draw[thick] (d) -- (i);
\draw[thick] (d) -- (j);
\end{tikzpicture}

%% file: appendix-proofs.tex
\section{Deferred proofs}
\label{sec:appendix-proofs}

\Fissubmodular*
\begin{proof}
For any arbitrary realization $\mathcal{H} \sim \mathcal{G}$ and negative node $u \in \mathbf{N}$, define indicator functions
\begin{equation}
\label{eq:g}
g_{\mathcal{H}, u}(\mathbf{S})
= \mathbbm{1} \left[ \text{$u$ is reachable from $\mathbf{P} \setminus \mathbf{S}$ in $\mathcal{H}[\mathbf{V} \setminus \mathbf{S}]$} \right]
\end{equation}
Then, we can rewrite $F(\mathbf{S})$ as follows:
\begin{align*}
F(\mathbf{S})
&= \mathbb{E}_{\mathcal{H} \sim \mathcal{G}} \left[ f_{\mathcal{H}}(\mathbf{S}) \right] \tag{By \cref{eq:F}}\\
&= \mathbb{E}_{\mathcal{H} \sim \mathcal{G}} \left[ \sum_{u \in \mathbf{N}} g_{\mathcal{H}, u}(\mathbf{S}) \right] \tag{By \cref{eq:realized-f} and \cref{eq:g}}
\end{align*}
Since non-negativity, monotonicity, and submodularity are preserved under non-negative linear combinations, it suffices to show that $g_{\mathcal{H},u}$ is non-negative, monotonically non-increasing, and submodular for every realization $\mathcal{H}$ and every $u \in \mathbf{N}$.

In the remainder of the proof, let us fix an arbitrary realization $\mathcal{H} \sim \mathcal{G}$, an arbitrary negative node $u \in \mathbf{N}$, arbitrary subsets $\mathbf{A} \subseteq \mathbf{B} \subseteq \mathbf{P}$, and an arbitrary positive node $x \in \mathbf{P} \setminus \mathbf{B}$.

\paragraph{Non-negativity and monotonically non-increasing.}
Since $g_{\mathcal{H},u}(\mathbf{S}) \in \{0,1\}$ is an indicator function, it is non-negative.
Furthermore, treating more positive nodes can only decrease this indicator function: since $\mathcal{H}[\mathbf{V} \setminus \mathbf{B}]$ is a subgraph of $\mathcal{H}[\mathbf{V} \setminus \mathbf{A}]$ and $\mathbf{P} \setminus \mathbf{B} \subseteq \mathbf{P} \setminus \mathbf{A}$, reachability can only decrease, so $g_{\mathcal{H},u}(\mathbf{A}) \geq g_{\mathcal{H},u}(\mathbf{B})$.

\paragraph{Submodularity.}
We need to establish
\[
g_{\mathcal{H},u}(\mathbf{A} \cup \{x\}) - g_{\mathcal{H},u}(\mathbf{A})
\geq g_{\mathcal{H},u}(\mathbf{B} \cup \{x\}) - g_{\mathcal{H},u}(\mathbf{B})
\]
Since $g_{\mathcal{H},u}$ is monotonically non-increasing, we always have
\[
g_{\mathcal{H},u}(\mathbf{A})
\geq g_{\mathcal{H},u}(\mathbf{B})
\geq g_{\mathcal{H},u}(\mathbf{B} \cup \{x\})
\quad
\text{and}
\quad
g_{\mathcal{H},u}(\mathbf{A})
\geq g_{\mathcal{H},u}(\mathbf{A} \cup \{x\})
\geq g_{\mathcal{H},u}(\mathbf{B} \cup \{x\})
\]
Since $g_{\mathcal{H},u}$ is $\{0,1\}$-valued, we see that $g_{\mathcal{H},u}(\mathbf{B} \cup \{x\}) - g_{\mathcal{H},u}(\mathbf{B}) \leq 0$, so the only nontrivial case to verify is that $g_{\mathcal{H},u}(\mathbf{B}) = 1$ and $g_{\mathcal{H},u}(\mathbf{B} \cup \{x\}) = 0$ when $g_{\mathcal{H},u}(\mathbf{A}) = 1$ and $g_{\mathcal{H},u}(\mathbf{A} \cup \{x\}) = 0$.

In this case, in $\mathcal{H}[\mathbf{V} \setminus \mathbf{A}]$, the node $u$ is reachable from some untreated positive, but after removing $x$, it is no longer reachable from any untreated positive.
Therefore, $x$ is the unique untreated positive from which $u$ is reachable in $\mathcal{H}[\mathbf{V} \setminus \mathbf{A}]$.
Hence, there exists a directed path from $x$ to $u$ in $\mathcal{H}[\mathbf{V} \setminus \mathbf{A}]$.
Since positive nodes have no incoming edges, all internal vertices on this path must be negative.
Therefore, removing additional positive nodes in $\mathbf{B} \setminus \mathbf{A}$ does not affect this path, and it remains present in $\mathcal{H}[\mathbf{V} \setminus \mathbf{B}]$.
Thus, $u$ is still reachable from $x$ in $\mathcal{H}[\mathbf{V} \setminus \mathbf{B}]$, implying that $g_{\mathcal{H},u}(\mathbf{B}) = 1$ and $g_{\mathcal{H},u}(\mathbf{B} \cup \{x\}) = 0$, verifying that $g_{\mathcal{H},u}$ is indeed submodular.
\end{proof}

Before we prove \cref{thm:approximation_ratio}, we first prove the remaining supporting lemma.


\algoanalysis*
\begin{proof}
For a fixed $\mathbf{S} \subseteq \mathbf{P}$, define random variables $X_i = f_{\mathcal{H}_i}(\mathbf{S})$, for $i \in \{1, \ldots, m\}$.
We see that these variables are independent and satisfy $0 \leq X_i \leq |\mathbf{N}|$.
Moreover, $\mathbb{E}[X_i] = F(\mathbf{S})$, and $\hat{F}_m(\mathbf{S}) = \frac{1}{m} \sum_{i=1}^m X_i$.
Applying Hoeffding's inequality (see \cref{thm:hoeffding}) yields
\[
\Pr \left( \left| F(\mathbf{S}) - \hat{F}_m(\mathbf{S}) \right|> \eps \right)
= \Pr \left( \left| m F(\mathbf{S}) - m \hat{F}_m(\mathbf{S}) \right| > m \eps \right)
\leq 2 \exp \left(-\frac{2m\eps^2}{|\mathbf{N}|^2} \right)
\]
Applying a union bound over all possible $2^{|\mathbf P|}$ subsets $\mathbf{S} \subseteq \mathbf{P}$ gives
\[
\max_{\mathbf{S} \subseteq \mathbf{P}} \Pr \left( \left| F(\mathbf{S}) - \hat{F}_m(\mathbf{S}) \right|> \eps \right) \leq 2^{|\mathbf P| + 1} \exp \left(-\frac{2m\eps^2}{|\mathbf{N}|^2} \right)
\]
The claim follows by plugging in $m$ as defined in the statement.
\end{proof}

We are now ready to prove \cref{thm:approximation_ratio}.

\approximationratio*
\begin{proof}
Define $\hat{F}_m$ as in \cref{alg:known}, where $\tilde{\mathbf{S}}$ is an optimizer.
Then, with probability at least $1 - \delta$,
\begin{align*}
F(\hat{\mathbf{S}})
&\leq \hat{F}_m(\hat{\mathbf{S}}) + \eps \tag{By \cref{lem:algo-analysis}}\\
&\leq 2 \sqrt{|\mathbf{P}|} \cdot \hat{F}_m(\tilde{\mathbf{S}}) + \eps \tag{By $2 \sqrt{|\mathbf{P}|}$-approximation \texttt{MkU} algorithm of \cite{chlamtac2018densest}}\\
&\leq 2 \sqrt{|\mathbf{P}|} \cdot \hat{F}_m(\mathbf{S}^\star) + \eps \tag{Since $\tilde{\mathbf{S}}$ is minimizer for $\hat{F}_m$}\\
&\leq 2 \sqrt{|\mathbf{P}|} \cdot (F(\mathbf{S}^\star) + \eps) + \eps \tag{By \cref{lem:algo-analysis}}\\
&\leq 2 \sqrt{|\mathbf{P}|} \cdot F(\mathbf{S}^\star) + \eps \cdot (1 + 2 \sqrt{|\mathbf{P}|})
\end{align*}
As each invocation of \cref{lem:algo-analysis} fails with probability at most $\delta$, union bound tells us that the total failure probability is at most $2 \delta$.

Since the MkU algorithm from \cite{chlamtac2018densest} runs in polynomial time and additional steps are polynomial time operations, \texttt{CAST} also runs in polynomial time.
\end{proof}

%% file: appendix-experiments.tex
\section{Implementation Details}
\label{sec:parameters}
\paragraph{Dataset Details.}
 We examine the HIV, Syphilis, and Hepatitis networks from the ICPSR dataset. The de-identified, public-use dataset can be downloaded at \url{https://www.icpsr.umich.edu/web/ICPSR/studies/22140} after agreeing to the terms and conditions. These networks contain 778, 542, and 1732 total nodes, of which 88, 44, and 117 are initially infected, respectively. Each node is associated with a covariate vector with 72 binary dimensions that are one-hot encoded from 17 categorical variables related to gender, sexual orientation, occupation, etc.

 \paragraph{Link Prediction for Imperfect Networks.}
 For our link prediction model $\ell$, we use an SVM with an RBF kernel \cite{scikit-learn} trained on the other disease network datasets from the ICPSR dataset. We pass it the cosine similarity between node covariates, the status of each node, an indicator of whether the two nodes share a neighbor, and the size of their shared connected component as features. We only predict edges between nodes within the same connected component, since the connected components are established in the given forest subgraph. In practice, keeping all predicted edges in $\mathcal{G}^\prime$ turns each connected component into a complete graph, so we drop any edges with a predicted probability below a threshold $p_\ell$. We use $p_\ell = 0.25$ in our experiments. 

\section{Imperfect Edge Probabilities}
\label{sec:edge-probs}
For these experiments, the decision graph $\mathcal{G}^\prime$ has the same topology as the true evaluation graph $\mathcal G$. However, edge probabilities in $\mathcal{G}^\prime$ are drawn from $\mathcal{P}^\prime$, an independent and separate distribution from $\mathcal{P}$, which is used to generate edge probabilities in $\mathcal{G}$. In the first case that we study, edge probabilities in $\mathcal{G}^\prime$ are drawn independently but from the same distribution as those in $\mathcal{G}$. In the second case, edge probabilities in $\mathcal{G}^\prime$ are drawn independently and from a completely different distribution. We use $\mathcal{P}^\prime = N(0.75, 0.2^2), \mathcal{P} = N(0.75, 0.2^2)$ in the first case and $\mathcal{P}^\prime = N(0.75, 0.2^2), \mathcal{P} = U(0, 1)$ in the second case. The case where the true edge probabilities are known to the algorithm, i.e., $\mathcal{P}^\prime = \mathcal{P} = N(0.75, 0.2^2)$ is also evaluated for reference. The results across various intervention budgets on the HIV network are shown in \cref{fig:edge_distributions}.

\begin{figure}
    \centering
    \includegraphics[width=\linewidth]{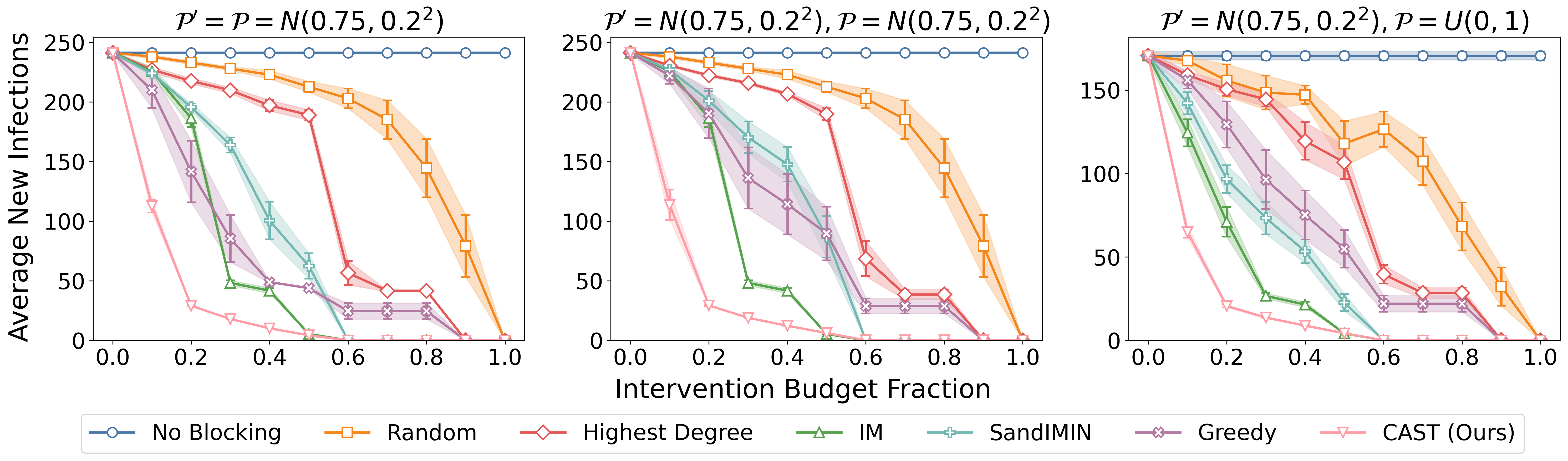}
    \caption{Average new infections ($\pm$ SE) across varying intervention budgets with imperfect edge probabilities. First, for reference, we show the case where true edge probabilities are known, i.e., $\mathcal{P}^\prime = \mathcal{P} = N(0.75, 0.2^2)$ (left). Then, we study the case where the distribution is known but the decision edge probabilities in $\mathcal{G}^\prime$ are drawn independently from those in the evaluation graph $\mathcal G$, i.e., $\mathcal{P}^\prime = N(0.75, 0.2^2), \mathcal{P} = N(0.75, 0.2^2)$ (middle). Finally, we study the case where the assumed edge probability distribution is mismatched with the true edge probability distribution, i.e., $\mathcal{P}^\prime = N(0.75, 0.2^2), \mathcal{P} = U(0, 1)$ (right).}
    \label{fig:edge_distributions}
\end{figure}

We observe that empirically, the performance of \texttt{CAST} remains stable and superior to baseline algorithms in the case where the distribution is known but independently sampled, i.e., $\mathcal{P}^\prime = N(0.75, 0.2^2), \mathcal{P} = N(0.75, 0.2^2)$. For example, under the 10\% and 20\% budgets, \texttt{CAST} only incurs 0.879\% and 0.875\% more average new infections (respectively), compared to the case where the true edge probabilities are known. For the case of $\mathcal{P}^\prime = N(0.75, 0.2^2), \mathcal{P} = U(0, 1)$, we can compare to the HIV plot in \cref{fig:core-result}, in which the same true edge probabilities were known to algorithms, i.e., $\mathcal{P}^\prime = \mathcal{P} = U(0, 1)$. Once again, \texttt{CAST} has the best performance across all algorithms. At 10\% and 20\% budgets, \texttt{CAST} only incurs 4.60\% and 0.430\% more average new infections (respectively), compared to the case where the true edge probabilities are known from \cref{fig:core-result}.

%% file: appendix-discussion.tex
\section{Discussion}
\label{sec:discussion}

\paragraph{Broader impact and privacy.}
As public health budgets face increasing strain \cite{ten2025impact}, resource-efficient strategies are essential.
\texttt{CAST} aligns with WHO network-based intervention guidelines \cite{who}, offering a mathematically rigorous tool for maximizing the impact of every intervention.
However, the sensitive nature of HIV status and transmission data necessitates robust privacy guardrails.
Consequently, any practical application of \texttt{CAST} must employ anonymization and secure computational environments to protect these sensitive pieces of information.

\paragraph{Ethics and fairness.} 
The work that we present here is limited to retrospective simulation studies on disease networks collected in the real world. 
There are no human subjects involved.
Before moving beyond simulations, we intend to conduct rigorous testing through pilot studies and field trials to observe algorithmic resource allocation. These trials will require IRB approvals at all academic institutions involved, and careful, informed consent will be sought from any participants.
We are working very closely with collaborators who have decades of experience in HIV prevention in the field in South Africa, are familiar with real-world resource allocation challenges, and will ultimately lead any algorithmic trials.
There is one potential ethical consideration that may arise due to the utilitarian objective of \texttt{MINT}: by prioritizing the prevention of large-scale infection cascades, one may naturally de-prioritize individuals in smaller, more isolated components.
While this maximizes the total number of infections averted, it may inadvertently introduce geographic or demographic disparities in resource access. 
All of the above considerations will be carefully weighed and discussed with our collaborators.
We will also proceed with appropriate IRB oversight across the different institutions involved.

\paragraph{Limitations and future work.}
Our theoretical guarantees and empirical findings are well-suited to our target setting with a one-shot intervention.
However, one limitation of the \texttt{MINT} model is that it does not capture sequential interventions in which agencies can dynamically update their strategies as budget constraints evolve and the transmission network changes.
Future work could focus on a new model for these adaptive problem instances.
SI may be a more appropriate model of spread for such problems, since it is often used in problems with multiple rounds of interventions, e.g., \cite{ou2021active}.
Additionally, investigating fairness-constrained variants of \texttt{MINT} that balance global infection reduction with equitable coverage is also an important direction for future study.
